\documentclass[12pt,twoside,a4paper]{article}

\usepackage{amsmath,amsthm,amsfonts,amssymb}
\usepackage{mathtools}
\usepackage{tikz}
\usetikzlibrary{arrows,automata}
\usepackage[latin1]{inputenc}
\usepackage{booktabs,lipsum}
\usepackage{tabularx}
\usepackage{float}

\begin{document}

\title{KdV Equation for Theta Functions on Non-commutative Tori}
\author{Independent Research for Honors in Mathematics\\
Supervisor: Prof. Emma Previato\\
Author: Wanli Cheng}
\date{October 5,2019}
\maketitle

\theoremstyle{definition} \newtheorem{Def}{Definition}[section]
\theoremstyle{definition} \newtheorem{Thm}{Theorem}[section]
\theoremstyle{definition} \newtheorem{Cor}[Thm]{Corrollary}
\theoremstyle{definition} \newtheorem{Con}[Thm]{Conjecture}
\theoremstyle{definition} \newtheorem{Lem}{Lemma}[section]
\theoremstyle{definition} \newtheorem{Not}{Notation}[section]
\theoremstyle{definition} \newtheorem{Prop}{Proposition}[section]
\theoremstyle{remark} \newtheorem*{Rem}{Remark}
\theoremstyle{remark} \newtheorem{Eg}{Example}[section]

\abstract
In the fields of non-commutative geometry and string theory, quantum tori
appear in different mathematical and physical contexts.
Therefore, quantized theta functions defined on quantum tori are
also studied (Yu. I. Manin, A. Schwartz; note that a comparison between the
two definitions of quantum theta is still an open problem). One important application of
classical theta functions is in soliton theory. Certain soliton equations, including
the KdV equation,  have algebro-geometric solutions that are given by theta
functions (we refer to F. Gesztesy and H. Holden), and as such belong to an
``integrable hierarchy''. While quantized integrability is a very active and
complicated subject, in this work we take a different, na\"{\i}ve approach.
We conduct an experiment: using a definition of differentiation on quantum
tori  (M. Rieffel), we ask whether the quantum theta function satisfies a
non-linear PDE. The experiment in successful on the 2-torus and for the KdV
equation. This opens the way to future investigations, such as the quest for a
compatible hierarchy satisfied by quantum theta, and a consistent definition of
complete integrability.

\newpage
\tableofcontents
\newpage

\section{Introduction}
In this paper, we study the relation between KdV equations \footnote{Here the quantized KdV means the smooth function $u(x, t)$ is replace by a state vector $\psi(x, t)$ takes value in $S(\mathbb{R})$, which is not the case that quantizing the Poisson structure of KdV.}
 and the theta vector $\vartheta$ defined by Schwarz [S], [E] on quantized 2-tori. KdV equations are a family of soliton equations of which the algebro-geometric solutions can be expressed by theta functions. In the classical case, if theta function satisfied the Hirota bilinear relation, then $2\frac{d^2}{dx^2}\log{\vartheta}(\alpha{x}+\gamma{t}+x_0, \tau)$ gives a family of soliton solution of the KdV equation. Similar to the classical case, we found that Schwartz's theta vector satisfies the Hirota bilinear relation and it is possible that we make sense of $2\frac{d^2}{dx^2}\log{\vartheta}(ax+bt+c, \tau)$ gives a family of solutions of KdV equation.\\
 \ \\
In This paper, we first introduce the classical theta function in Section 2.  In Section 3, we review the concepts of $C^*$-algebra which gives the basic setting of non-commutative geometry. Then, in Section 4, we study the quantized tori in the approach of Rieffel [R]. Quantized tori are considered as $C^*$-algebras formed by deformation quantization of ordinary tori. Therefore the definitions and theorems of $C^*$-algebra are used to depict the structure of quantized tori. This section establishes the basic framework of our research. Then in Section 5, following [S] and [E], we construct the holomorphic structure and the simplest theta vectors on quantized tori. In Section 6, we reproduce the classical relation between KdV equation and (classical) theta functions. Finally in Section 7, we calculate the derivatives of theta vectors and show that these quantum theta vectors satisfies Hirota's bilinear relation.
 
\section{Classical Theta functions}
In this section we study the \emph{classical theta functions} which are ubiquitous in algebra geometry and number theory.  Theta functions are applied in the soliton theory. The soliton solutions of integrable systems in KdV hierarchy or AKNS hierarchy can be represented by theta functions [GH]. In this paper, we focus on the link between theta functions and KdV equations, hence our scope is limited to the Jacobi-related theta functions. In subsection 2.1 we define the classical theta functions in a comprehensive way. Then in subsection 2.2 we introduce the Jacobi varieties and Abel maps to prepare us with the means relating theta functions and algebraic curves. In next section, this relation would be discussed.

\subsection{Classical Theta Functions}
Here we start from the \emph{theta functions} with one complex variable.

\begin{Def}[Holomorphic function]
\ \\
A \emph{holomorphic function} is a complex valued function with one or more variables that is differentiable in its neighborhood of every point in its domain. A holomorphic function whose domain is the whole domain is called an \emph{entire function}.
\end{Def}

\begin{Def}[Theta function]
\ \\
The (\emph{single variable}) \emph{theta function} is defined as:
\begin{equation}
\vartheta:\mathbb{C}\times\mathbb{H}\to\mathbb{C},
\end{equation}
such that $\forall{z}\in\mathbb{C}$ and $\forall{\tau}\in\mathbb{H}$,
\begin{align}
\vartheta(z, \tau)&=\sum_{n\in\mathbb{Z}}\textnormal{exp}(\pi{i}n^2\tau+\pi{i}nz),\\
&=1+2\sum_{n=1}^{\infty}(e^{\pi{i}\tau})^{n^2}\textnormal{cos}(2\pi{nz}),\\
&=\sum_{n\in\mathbb{Z}}q^{n^2}\eta^n,
\end{align}
where $q=e^{\pi{i}\tau}$ and $\eta=e^{2\pi{i}z}$.
\label{theta}
\end{Def}

\noindent If $\tau$ is fixed, the theta function $\vartheta(z, \tau)$ can be viewed as a Fourier series for a periodic \emph{entire function} of $z$ with period 1. Then the theta function would satisfy the identity:
\begin{equation}
\forall{m}\in\mathbb{Z}\quad\vartheta(z+m, \tau)=\vartheta(z, \tau).
\end{equation}
And respect the quasi-period $\tau$, the theta function also satisfies an equation:
\begin{equation}
\forall{m}\in\mathbb{Z}\quad\vartheta(z+m\tau, \tau)=\textnormal{exp}(-\pi{i}m^2\tau-2\pi{m}z)\vartheta(z, \tau),
\end{equation}

\noindent Then, to study the cases with genus more than one, we extend our definition by introducing the \emph{Siegel upper half-plan}.

\begin{Def}[Siegel's upper half-plane]
The Siegel's upper half-plane $\mathbb{H}^g$ is an generalization of upper-half space $\mathbb{H}$ with genus $g\ge1$. It is defined as:
\begin{equation}
\mathbb{H}^g:=\{\tau\in{\textnormal{Sym}_g(\mathbb{C})}|\forall\hat{z}\in\mathbb{C}^g(\hat{z}\ne0\Rightarrow{\hat{z}^{T}\textnormal{Im}({\tau})\hat{z}>0})\},
\end{equation}
 where $\textnormal{Sym}_n(\mathbb{C})$ is use to denote the set of symmetric $g\times{g}$ matrices with complex entries and $\hat{z}^T$ the transpose of $\hat{z}$.
\end{Def}

\noindent Now we can define the multivariable theta function called \emph{Riemann theta function} analogically.

\begin{Def}[Riemann Theta function]
\ \\
The \emph{Riemann theta function} is defined as:
\begin{equation}
\vartheta:\mathbb{C}^g\times\mathbb{H}^g\to\mathbb{C},
\end{equation}
such that $\forall\vec{z}\in\mathbb{C}^g$ and $\forall\tau\in\mathbb{H}^g$,
\begin{equation}
\vartheta(\vec{z}, \tau)=\sum_{\vec{n}\in\mathbb{Z}^{g}}\textnormal{exp}(\pi{i}\vec{n}^{T}\tau\hat{n}+2\pi{i}\vec{n}^{T}z).
\end{equation}
And similarly, the Riemann theta function is a \emph{holomorphic function} on $\mathbb{C}^g\times\mathbb{H}^g$ with the following identities hold:
\begin{align}
&\forall\vec{w}\in\mathbb{Z}^g&&\vartheta(\vec{z}+\vec{w}, \tau)=\vartheta(\vec{z}, \tau),\\
&\forall\vec{w}\in\mathbb{Z}^g&&\vartheta(\vec{z}+\vec{w}, \tau)=\textnormal{exp}(-\pi{i}\vec{w}^{T}\tau\vec{w}-2\pi{i}\vec{w}^{T}\vec{z})\vartheta(\vec{z},\tau).
\end{align}
Follows from (10), (11), $\vartheta(\vec{z}, \tau)$ can be also considered as a section of the line bundle over the torus $\mathbb{C}^g/\mathbb{Z}^{2g}$.
\label{R-theta}
\end{Def}

\subsection{Hyperelliptic curves}

Then we would introduce the hyperelliptic curves. Actually, for each hyperelliptic curve, we can associate a theta function with it.

\begin{Def}[Affine Plane Complex Algebraic Curve]
\ \\
An \emph{affine plane complex curve} $\mathcal{K}$ is the locus of zeros of a non-constant polynomial $\mathcal{F}$ in $\mathbb{C}^2$. If we have:
\[
\nabla\mathcal{F}(z_0,y_0)=(\mathcal{F}_z(z_0,y_0), \mathcal{F}_z(z_0,y_0))\ne0,
\]
then the polynomial $\mathcal{F}$ is called \emph{nonsingular} at a root $(z_0, y_0)$. Hence the algebraic curve is called \emph{nonsingular} at the point $P_0=(z_0, y_0)$.
\end{Def}

\begin{Thm}
\ \\
For $\mathcal{K}$ is an affine algebraic curve, if the polynomial $\mathcal{F}$ associate with it is both nonsingular and irreducible, then $\mathcal{K}$ is a Riemann surface.
\end{Thm}

\begin{proof}
This theorem can be proved directly from the definition of a Riemann surface.  The non-singularity of the polynomial $\mathcal{F}$ ensured the affine curve $\mathcal{K}$ associated with it is smooth and the irreducibility of $\mathcal{F}$ ensured the connectedness of $\mathcal{K}$.
\end{proof}

\begin{Def}[Hyperelliptic compact Riemann surface]
\ \\
A compact Riemann surface is called hyperelliptic if it admits a meromorphic function of degree 2, i.e., a non-constant meromorphic function with precisely two poles counting multiplicity.
\end{Def}

\subsection{Homology basis and Period Lattice}

Here we introduce the homology basis and period relations of an algebraic curve. By getting use this period, we associate an algebraic curve with a Riemann theta function as well as a period lattice which can induce a Jacobi variety. 

\begin{Def}[Homology Basis of an Algebraic Curve]
\ \\
Given an algebraic curve $\mathcal{K}_g$. Suppose $\mathcal{K}_g$ is a compact Riemann surface, of genus $g\in\mathbb{Z}$, then we can choose a homology basis $\{a_j, b_j\}_{j=1}^g$ of $\mathcal{K}_g$ such that the cycles $a_j$ and $b_k$ satisfies the following identities:
\begin{align}
&a_j\circ{b_k}=\delta_{j, k},&&a_j\circ{a_k}=0,&&b_j\circ{b_k}=0&&j, k=1,\dots,g,
\end{align}
where these cycles intersecting and forms a right-handed coordinate system.
\end{Def}

\begin{Thm}[Riemann's Period Relation]
\ \\
Given an algebraic curve $\mathcal{K}_g$ with a homology basis $\{a_i, b_j\}$. Then suppose $\omega$ and $\nu$ are closed $C^1$ meromorphic differentials (1-forms) on $\mathcal{K}_g$, we have:
\begin{equation}
\iint\limits_{\mathcal{K}_g}\omega\wedge\nu=\sum_{j=1}^{g}((\int_{a_j}\omega)(\int_{b_j}\nu)-(\int_{b_j}\omega)(\int_{a_j}\nu)).
\end{equation}
Then, if $\omega$ and $\nu$ are holomorphic 1-forms on $\mathcal{K}_g$, we have:
 \begin{equation}
 \sum_{j=1}^{g}((\int_{a_j}\omega)(\int_{b_j}\nu)-(\int_{b_j}\omega)(\int_{a_j}\nu))=0.
 \end{equation}
In addition, if $\omega$ is a nonzero holomorhphic 1-form on $\mathcal{K}_g$, then
\begin{equation}
\text{Im}(\sum_{j=1}^{g}(\int_{a_j}\omega)(\int_{b_j}\omega))>0
\end{equation}
\end{Thm}

\begin{proof}
This result can be obtained by Stokes' theorem and a canonical dissection of $\mathcal{K}_g$. The dissection is along the cycles of $\mathcal{K}_g$ and yields the simply connected interior $\hat{\mathcal{K}}_g$ of the fundamental polygon $\partial\hat{\mathcal{K}}_g$, which is given by
\begin{equation*}
\partial\hat{\mathcal{K}}_g=a_1b_1a^{-1}_1b^{-1}_1a_2b_2a^{-1}_2b^{-1}_2,\dots,a^{-1}_gb^{-1}_g.
\end{equation*}
\end{proof}

\begin{Def}[b-periods of differentials of the first kind]
\ \\
Given the homology basis $\{a_i, b_j\}$, we can then define the normalized basis of the spaces of holomorphic differentials on $\mathcal{K}_g$ such that,
\begin{equation}
\int_{a_k}\omega_j=\delta_{j, k},\quad j, k = 1,\dots,g.
\end{equation}
Then the b-period of $\omega_j$ is defined as:
\begin{equation}
\tau_{j, k}=\int_{b_k}\omega_j,\quad j, k= 1,\dots,g 
\end{equation}
\end{Def}

\begin{Cor}
\ \\
The b-periods of differentials of the first kind $\tau$ is an element of $\mathbb{H}^g$.
\end{Cor}

\begin{proof}
\ \\
Given Theorem 1.2 and the definition of $\mathbb{H}^g$, this corollary immediately follows.
\end{proof}

\begin{Def}[Periodic Lattice]
\ \\
The \emph{periodic lattice} $\Lambda_g$ in $\mathbb{C}^g$ is given by:
\begin{equation}
\Lambda_g=\{\vec{z}\in\mathbb{C}^g|(\exists\vec{n}, \vec{m}\in\mathbb{Z}^g)\ \vec{z}=\vec{n}+\vec{m}\tau\},
\end{equation}
where $\tau$ is an element of $\mathbb{H}^g$. Given an algebraic curve $\mathcal{K}^g$, by Corollary 1.3, we have a $\tau\in\mathbb{H}^g$ associate with this curve. Hence there is also an periodic lattice $\Lambda_g$ generated by this $\tau$ associate to the curve $\mathcal{K}_g$. 
\end{Def}

\begin{Def}[Theta function of a compact Riemann surface]
\ \\
Given a compact Riemann surface $\mathcal{K}_g$, with a homology basis $\{a_j, b_j\}_{j=1}^g$ and $\tau$, the b-periods of the differentials of the first kind, $\{\omega_j\}_{j=1}^g$. The Riemann theta function associated with $\mathcal{K}_g$ and the homology basis naturally arises as:
\begin{equation}
\vartheta(\vec{z})=\sum_{\vec{n}\in\mathbb{Z}^{g}}\textnormal{exp}(\pi{i}\vec{n}^{T}\tau\vec{n}+2\pi{i}\vec{n}^{T}z),\quad\vec{z}\in\mathbb{C}^g.
\end{equation}
By Corollary 1.X, we have $\tau\in\mathbb{H}^g$, hence $\vartheta(\vec{z}, \tau)$ is a well defined entire function on $\mathbb{C}^g$ satisfying the algebraic properties (10) and (11).
\end{Def}

\subsection{Jacobi Variety and Abel Map}

Here we introduce the Jacobi variety. Then we would study the Abel map and Jacobi inverse theory which allow us to interplay between the compact Riemann surface $\mathcal{K}_g$ and the its Jacobi variety $J(\mathcal{K}_g)$.

\begin{Def}[Jacobi variety]
\ \\
Given an algebraic curve $\mathcal{K}_g$, we then have its quasi-period $\tau$ and the period lattice $\Lambda_g$, then the Jacobi variety of a curve (which is an algebraic version of the Jacobian of a curve) can be defined. The Jacobi variety $J(\mathcal{K}_g)$ of $\mathcal{K}_g$ is defined as:
\begin{equation}
J(\mathcal{K}_g)=\mathbb{C}^g/\Lambda_g.
\end{equation}
It is a abelian variety and obey a symmetry called Jacobi group.
\end{Def}

\begin{Def}[Abel-Jacobi map]
\ \\
There is a map relating an algebraic curve $\mathcal{K}_g$ with its Jacobi variety $J(\mathcal{K}_g)$ called the Abel-Jacobi map. It is given by:
\begin{equation}
\forall{Q_0}\in\mathcal{K}_g,\quad\hat{A}_{Q_0}:\mathcal{K}_g\to{J(\mathcal{K}_g)};
\end{equation}
such that $\forall{P}\in\mathcal{K}_g$,
\begin{align}
\hat{A}_{Q_0}(P)&=(A_{Q_0, 1}(P),\dots, A_{Q_0, g}(P))\\
&=(\int_{Q_0}^{P}\omega_1,\dots,\int_{Q_0}^{P}\omega_g)\quad(\textnormal{mod}\ \Lambda_g)
\end{align}
\end{Def}

\section{$C^*$-Algebra}

$C^*$-algebras play an important role in the study of non-commutative geometry. There is an equivalence between the $C^*$-algebras and the Hausdorff spaces. Then, to study the non-commutative space, we first consider a non-commutative algebra and consider additional structure on it. In this section, we introduce the basic definitions of $C^*$-algebra which would be used to construct a quantized tori as a non-commutative differential manifold.

\subsection{Linear Algebra}

Here, we introduce the definitions of the Banach and the Hilbert spaces which would be used to define a $C^*$-algebra later.

\begin{Def}[Norm]
\ \\
A \emph{norm} on a (complex) vector space $V$ is a map $||\ ||:V\to\mathbb{R}$ such that,
\begin{align}
&||v||\le0;\\
&||v||=0\iff{v}=0;\\
&||\lambda{v}||=|\lambda|\ ||v||;\\
&||v+w||\le||v||+||w||&(\text{Triangle Inequality});
\end{align}
forall $v, w\in{V}$, $\lambda\in\mathbb{C}$. A vector space with a norm on it is called a \emph{normed vector space}. And a norm on $V$ defines its associate \emph{metric} $d$ on $V$ by $d(v, w):=||v-w||$. Hence a normed vector space is a \emph{metric space} under the associate metric of its norm.
\end{Def}

\begin{Def}[Inner Product]
\ \\
An \emph{Inner product} on a vector space $V$ is a map $(,):V\times{V}\to\mathbb{C}$ such that,
\begin{align}
&\overline{(v, w)}=(w, v)\\
&(v, \lambda{w}+\mu{u})=\lambda(v, w)+\mu(v, u)\\
&(v, v)\ge0\\
&(v, v)=0\iff{v=0}
\end{align}
for all $v, w, u\in{V}$ and $\lambda, \mu\in\mathbb{C}$. A vector space with an inner product defined on it is called an \emph{inner product space}. An inner product on $V$ also defines a norm on $V$ by $||v||=\sqrt{(v, v)}$. Hence $V$ is also a \emph{metric space} with respect to the metric associated with this norm. Then we can derive the \emph{Cauchy-Schwartz Inequality} on an inner product space such that
\begin{equation}
|(u, v)|\le||v||\cdot||w||
\end{equation}
for all $v, w\in{V}$.
\end{Def}

\noindent Given the norm and inner product defined, we can then define the Banach spaces and Hilbert spaces.

\begin{Def}[Completeness]
\ \\
A vector space is called \emph{completed} in some metric when all the Cauchy sequences respect to this metric converge in this space.
\end{Def}

\begin{Def}[Banach Space]
\ \\
A \emph{Banach space} is a normed vector space which is complete in its associate metric.
\end{Def}

\begin{Def}[Hilbert Space]
\ \\
A \emph{Hilbert space} is an inner product space which is complete in its associate metric.
\end{Def}

\begin{Def}[Bounded Operators]
\ \\
A \emph{bounded} (\emph{linear}) \emph{operator} on a \emph{Banach space} $B$ is a linear map $\beta:B\to{B}$ such that
\begin{equation}
||\beta||:=\text{sup}\{||\beta{v}||\ |\ v\in{B}\ \text{and}\ ||v||=1\}<\infty,
\end{equation}
where this norm is called the \emph{operator norm} on $B$. When $B$ is in addition a \emph{Hilbert space} $H$, then the norm is given by
\begin{equation}
||\beta||:=\text{sup}\{\sqrt{(\beta\psi,\beta\psi)}\ |\ \psi\in{H}\ \text{and}\ (\psi, \psi)=1\}.
\end{equation}
When $\beta$ is bounded, we then have
\begin{equation}
||\beta{v}||\le||\beta||\ ||v||.
\end{equation} 
\end{Def}

\begin{Eg}
Given a Banach space $B$, then the space $\mathfrak{B}(B)$ of bounded operators on $B$ is itself a Banach space in the operator norm. The proof directly follows the definitions.
\end{Eg}

\begin{Def}[Self-Adjoint operators]
Given a Hilbert space $H$, a \emph{self-adjoint operator} $A$ on $H$ is a bounded linear operator $A:H\to{H}$ such that
\begin{equation}
A^*=A
\end{equation}
where $*$ is the \emph{operator adjoint} defined by
\begin{equation}
(A\psi, \varphi)=(\psi, A^*\varphi).
\end{equation}
\end{Def}

\begin{Def}[Unitary Operators]
\ \\
Given a Hilbert space $H$, a \emph{Unitary operator} $U$ on $H$ is a bounded linear operator $U:H\to{H}$ such that
\begin{equation}
U^*{U}=UU^*=I
\end{equation}
where $I:H\to{H}$ is the identity operator and $U^*$ is the \emph{adjoint} of $U$.
\end{Def}

\subsection{Banach Algebra}
In this subsection, we give the definition and an example of Banach algebras. We would then define $C^*$-algebras as Banach algebras with involutions in next subsection.

\begin{Def}[Algebra]
\ \\
An \emph{algebra} is a vector space $\mathfrak{A}$ with an associative bilinear operation $\cdot:\mathfrak{A}\times\mathfrak{A}\to\mathfrak{A}$ usually called ``multiplication". And $\alpha\cdot\beta$ is usually written as $\alpha\beta$. If $\alpha\beta=\beta\alpha$, we say the algebra $\mathfrak{A}$ is \emph{commutative}.
\end{Def}

\begin{Def}[Banach Algebra]
\ \\
A \emph{Banach algebra} is an algebra $\mathfrak{A}$ which is also a Banach space, which satisfies:
\begin{equation}
\forall{\alpha, \beta}\in\mathfrak{A},\quad||\alpha\beta||\le||\alpha||\ ||\beta||.
\end{equation}
\end{Def}

\begin{Eg}
Let $B$ denote a Banach space, then $\mathcal{B}(B)$, the set of bounded operators on $B$ forms a Banach algebra under the operator norm.
\end{Eg}

\subsection{$C^*$-Algebra}
Here we define the $C^*$-Algebras, and we reproduce the theorem proved by Gelfund and Naimark which connected the $C^*$-algebra and geometry by presenting an equivalence between $C^*$-algebras and Hausdorff spaces. Via this intuition, we can study non-commutative geometry only with the corresponding (non-commutative) $C^*$-algebra of the topological spaces.

\begin{Def}[Involution]
\ \\
An \emph{involution} on an algebra $\mathfrak{A}$ is a real-linear map $\alpha\mapsto{\alpha^*}$ such that
\begin{align}
\alpha^{**}&=\alpha\\
(\alpha\beta)^{*}&=\beta^*\alpha^*\\
(\lambda\alpha)^*&=\overline{\lambda}\alpha^*
\end{align}
for all $\alpha, \beta\in\mathfrak{A}$. A $^*$-algebra is an algebra with an involution.
\end{Def}

\begin{Eg}
Let $H$ denote a Hilbert space, then $\mathcal{B}(H)$, the set of bounded operators on $H$ forms a Banach $^*$-algebra by defining an involution on it as the \emph{operator adjoint}.
\end{Eg}

\begin{Def}[$C^*$-Algebra]
\ \\
A $C^*$-\emph{algebra} is a $^*$-algebra $\mathfrak{A}$ which is at same time a complex Banach space, where
\begin{align}
||\alpha\beta||&\le||\alpha||\ ||\beta||\\
||\alpha^*\alpha||&={||\alpha||}^2
\end{align}
for all $\alpha\in\mathfrak{A}$.
\end{Def}

\begin{Cor}
\ \\
A Banach $^*$-algebra $\mathfrak{A}$ is a $C^*$-algebra if
\begin{equation}
||A||^2\le||A^*A||
\end{equation}
for all $A\in\mathfrak{A}$
\end{Cor}

\begin{Def}[Morphisms between $C^*$-Algebras]
\ \\
A \emph{morphism} between $C^*$-algebras $\mathfrak{A}, \mathfrak{B}$ is a complex-linear map $\phi:\mathfrak{A}\to\mathfrak{B}$such that,
\begin{align}
\phi(\alpha\beta)&=\phi(\beta)\phi(\alpha),\\
\phi(\alpha^*)&=\phi(\alpha)^*,
\end{align}
for all $\alpha\in\mathfrak{A}$, $\beta\in\mathfrak{B}$.
An {isomorphism} is a bijective morphism. Two $C^*$-algebra are \emph{isomorphic} if there exists a isomorphism between them. 
\end{Def}

\begin{Def}[Unit]
\ \\
A \emph{unit} in a Banach algebra $\mathfrak{A}$ is an element $\mathfrak{1}\in\mathfrak{A}$ such that
\begin{gather}
\mathfrak{1}\alpha=\alpha\mathfrak{1}=\alpha,\\
||\mathfrak{1}||=1.
\end{gather}
A \emph{unital} Banach algebra is a Banach algebra with unit.
\end{Def}

\noindent Then we introduce the celebrated result given by Gelfund and Naimark in 1943, which provide us an intuition associate the Hausdorff spaces and the $C^*$ algebras.

\begin{Thm}[Gelfand-Naimark]
\ \\
These mappings
\begin{align*}
\{\text{Locally compact Hausdorff spaces}\}&\to\{\text{Commutative}\ C^*\text{-algebras}\}\\
X&\mapsto{C}(X)\cong\mathfrak{C}_X
\\
&\text{and}
\\
\{\text{Compact Hausdorff spaces}\}&\to\{\text{Commutative unital}\ C^*\text{-algebras}\}\\
X&\mapsto{C}(X)\cong\mathfrak{C}_X
\end{align*}
are equivalences of categories. In words, for any compact/locally compact Hausdorff space $X$, the continuous functions $C(X)$, and the commutation algebra of $C(X)$ respect to the pointwise multiplication and addition forms a C*-algebra.\\
\end{Thm}

\begin{Def}[Non-commutative Space]
\ \\
The non-commutative spaces are defined by:
\begin{equation*}
\{\text{NC Compact Hausdorff spaces}\}:=\{\text{NC}\ C^*\text{-algebras}\ \text{and}\ C^*\text{-morphisms}\}^{\text{op}},
\end{equation*}
where ``$NC$" is the abbreviation for ``non-commutative".
\end{Def}

\noindent Obviously, this definition comes from the intuition given by Theorem 2.2. And in this convention, the non-commutative spaces are studied by non-commutative algebras, and the topological and geometrical structures of the space is also studied in the framework of $C^*$-algebra.

\subsection{$C^*$ modules}

Now, given a (compact/locally compact) Hausdorff space $X$, we want to find the counterparts of the vector bundles in the framework of the $C^*$-algebra given by $C(X)$. Hence we can study geometry on $X$ with structures and theorems of $C^*$-algebra. This correspondence plays a important role in non-commutative geometry. Since the non-commutative space is not defined explicitly, to define the bundles, connections, curvatures on it, we have to use the corresponding structures in $C^*$-algebra. In this subsection, we focused on the celebrated Serre-Swan theory, which gives an equivalence between the vector bundles and the $C^*$-modules. In non-commutative geometry, the vector bundles are then defined by its $C^*$-modules.

\begin{Def}[Vector Bundles]
\ \\
Let $M$ denote a topological space, then a (\emph{complex}) \emph{vector bundle} \emph{of rank} $k$ \emph{over} $M$ is a topological space $E$ together with a surjective map $\pi:E\to{M}$ such that:\\
(1)\quad$\forall{p\in{M}}$, $\pi^{-1}(p)$ is a vector space of dimension $k\in\mathbb{Z}^+$.\\
(2)\quad$\forall{p\in{M}}$, $\exists{U_p}\subset{M}$ where $U_p$ is a neighborhood of $p$ in $M$, and $\exists\Psi:\pi^{-1}(U):\to{U\times\mathbb{C}^k}$ which is a homeomorphism called the \emph{local trivialization of} $E$ \emph{over} $U$ for which
\begin{align*}
&(i)&&\pi_U\circ\Psi=\pi(\text{where}\ \pi_U:U\times\mathbb{C}^k\to{U}\ \text{is the projection}),\\
&(ii)&&\forall{q}\in{U},\ \text{the restriction of}\ \Psi \text{to}\ E_q\ \text{is a vector space isomorphism}\\
&&&\text{from}\ E_q\ \text{to}\ \{q\}\times\mathbb{C}^k\cong\mathbb{C}^k.
\end{align*}
Here $\pi^{-1}{p}$ is always called a \emph{fiber over} $p$. When both $E, M$ are smooth manifolds, $\pi$ is a smooth map and $\Psi$ are diffeomorphisms, then $E$ is called a \emph{smooth vector bundle}.
\end{Def}

\begin{Thm}
\ \\
Let $E$ denote a rank $k$ complex vector bundle over a Hausdorff space $M$. There is an integer $n\le{k}$ and an idempotent $p\in{C(X,{Mat}n(\mathbb{C}))X}$ such that $E\subseteq{X}\times\mathbb{C}^n$ and $\pi^{-1}(x)=p(x)\mathbb{C}^n$.
\label{Thm-Ep}
\end{Thm}

\begin{Def}[Sections]
\ \\
Given a (complex) smooth vector bundle $\pi:E\to{M}$ over a  smooth manifold $M$, a \emph{section} is a map $\Psi:M\to{E}$ such that
\begin{equation}
\pi(\Psi(x))=x\\
\end{equation}
for all $x\in{M}$. And the \emph{continuous sections} from $M$ to $E$ form a vector space under pointwise addition and scalar multiplication which is denoted as $\Gamma(E)$.
\end{Def}

\begin{Thm}
\ \\
When $M$ is a connected Hausdorff space, and $E$ a vector bundle over $M$, then $\Gamma(E)$ is a right-module for the commutative $C^*$-algebra $C(M)$. The linear right action $\beta$ of $C(M)$ on $\Gamma(E)$ is given by
\begin{equation}
\beta(f)(\Psi)(x):=f(x)\Psi(x)
\end{equation}
for all $f\in{C}(M)$, $\Psi\in\Gamma(E)$.
\end{Thm}

\begin{Def}[Free Module]
\ \\
A \emph{free module} $\Xi$ of an algebra $\mathfrak{A}$ is a direct sum
\begin{equation}
\Xi=\otimes^m\mathfrak{A}
\end{equation}
of a number of copies of $\mathfrak{A}$ itself. The linear action $\beta$ of $C(M)$ on $\Gamma(E)$ is then given by
\begin{equation}
\beta(B)(A_1\otimes\cdots\otimes{A}_m)(x):=A_1B\otimes\cdots\otimes{A_m}B
\end{equation}
for all $B\in\mathfrak{A}$, $A_1\otimes\cdots\otimes{A}_m\in\otimes^m\mathfrak{A}$. For $m$ is finite, the free module $\Xi=\otimes^m\mathfrak{A}$ is then called a \emph{finitely generated free module}.
\end{Def}

\begin{Cor}
\ \\
Consider the trivial bundle $M\times\mathbb{C}^n$ over the space $M$. We have
\begin{equation}
\Gamma(X\times\mathbb{C}^n)\cong(C)(M, \mathbb{C}^n)\cong{C(X)}\otimes\mathbb{C}^m\cong\otimes^mC(M).
\end{equation}
Hence, for finite $m$, $\Gamma(X\times\mathbb{C}^n)\cong\otimes^mC(M)$ is a \emph{finitely generated free module} for $C(X)$.
\end{Cor}

\begin{Def}[Projective Module]
\ \\
A \emph{projective module} $\Xi$ of an algebra $\mathfrak{A}$ is given by
\begin{equation}
\Xi=p\otimes^m\mathfrak{A}
\end{equation}
for all idempotent $p\in{Mat}_m(\mathfrak{A})$. Since the action of $p$ on $\otimes^n\mathfrak{A}$ commutes with the right multipication by $\mathfrak{A}$ on each component, it is a well-defined module. If $m$ is finite,  $\Xi=p\otimes^m\mathfrak{A}$ is called a \emph{finitely generated projective module}.
\end{Def}

\begin{Cor}
\ \\
Given a non-trivial rank $k$ complex vector bundle $E_p$ over a connected compact Hausdorff space $M$. Let $n\le{k}$, $p\in{C}(M, {Mat}_n(\mathbb{C}))$ and $E_p\in{X}\times\mathbb{C}^n$, with $\pi^{-1}(x)=p(x)\mathbb{C}^n$. Then we have
\begin{equation}
\Gamma(E_p)=p\otimes^n{C(M)}.
\label{Cor-Ep}
\end{equation}
For $n$ is finite, $\Gamma(E_p)=p\otimes^n{C(M)}$ is a \emph{finitely generated projective module}.
\end{Cor}

\begin{Thm}[Serre-Swan Theorem]
\ \\
Let $M$ denote a connected compact Hausdorff space. There is a bijective correspondence between complex vector bundles $E$ over $M$ and \emph{finitely generated projective modules } $\Xi(E)=\Gamma(E)$ for $C(M)$.
\end{Thm}

\begin{proof}
\ \\
This result follows Theorem \ref{Thm-Ep}, where any vector bundle $E$ is of the form $E_p$ hence correspond to a finitely generated projective module by (\ref{Cor-Ep}). Conversely, given a \emph{finitely generated projective module} $p\otimes^n{C(M)}$, one find such a $E_p\in{X}\times\mathbb{C}^n$ hence $E_p$ is defined.
\end{proof}

\begin{Def}[Hermitian Vector Bundle]
\ \\
In commutative case, a \emph{Hermitian vector bundle} $E$ is a complex vector bundle with an inner product ${\langle,\rangle}_x$ called \emph{Hermitian metric} defined on each fiber $\phi^{-1}(x)$, where  $\langle,\rangle_x$ is continuously depends on $x$, i.e.,
\begin{equation}
{C(X)}\ni\langle,\rangle_x:x\to\langle\psi(x),\varphi(x)\rangle_x
\end{equation}
for all $\psi, \varphi\in\Gamma(V)$.
\end{Def}

\section{Quantized Tori}
In this section, we would study quantized torus as a non-commutative differentiable manifold. Quantized tori arised in many different areas and actually, can be defined in various ways. In this paper, we followed [R] and view the quantized tori as a deformation quantization of classical tori. The method of deformation quantization developed by Kontsevich is introduced and we then use it to construct quantized tori. In addition, we would study the topological and smooth structures on the quantized tori.

\subsection{Deformation Quantization}

Physicists use (classical and quantum) mechanics to study the time evolution of a system. In classical mechanics, the possible states of a system forms a Poisson manifold $M$ and the observables (the physical quantities can be measured) are the smooth functions $C^\infty(M)$, which forms a commutative algebra. But in quantum mechanics, the possible states of a system are vectors in a Hilbert space and the observables are given by self-adjoint (Hermitian) operators, forming a $C^*$-algebra. In detail, Deformation quantization defines an algebra of observables in the quantum object corresponding to the observables of classical object $C^\infty(M)$ by deforming the poisson structure on the commutative algebra of $C^\infty(M)$.\\

\begin{Def}[Lie Algebra]
\ \\
A \emph{Lie algebra} is a vector space $\mathfrak{g}$ over some field $K$ with a binary relation $[,]:\mathfrak{g}\to\mathfrak{g}$called the \emph{Lie bracket} such that
\begin{align*}
&[ax+by, z]=a[x, z]+b[y, z], [z, ax+by]=a[z, x]+b[z, y];&(\text{Bilinearity})\\
&[x, x]=0;&(\text{Alternativity})\\
&[x, [y, z]]+[z, [x, y]]+[y, [z, x]]=0;&(\text{Jacobi identity})\\
&[x, y]=-[y, x].&(\text{Anticommutativity})
\end{align*}
\end{Def}

\begin{Def}[Poisson Algebra]
\ \\
A \emph{Poisson algebra} is a real vector space $A$ equipped with a commutative associate algebra structure
\begin{equation*}
(f, g)=fg,
\end{equation*}
with a Lie algebra structure
\begin{equation}
(f, g)\to\{f, g\}
\end{equation}
satisfying the compatibility condition
\begin{equation}
\{fg, h\}=f\{g, h\}+\{f, h\}g.
\end{equation}
\end{Def}

\begin{Def}[Poisson Manifold]
\ \\
A \emph{Poisson Manifold} is a manifold $M$ such that the commutative algebra of $C^\infty(M)$ is a Poisson algebra where the commutative product is defined by the ordinary pointwise multiplication of smooth functions and a Lie product satisfying Leibniz rule.
\end{Def}

\begin{Def}[Formal Deformation of Poisson Algebra]
\ \\
A \emph{formal deformation} of an algebra $A=C^\infty(M)$, is defined by a \emph{star-product} $*$, which is a map
\begin{align*}
*_\hslash:A\times{A}&\to{A}[\hslash]\\
(f, g)&\mapsto\sum_{k=0}^{\infty}f\times_kg\hslash^k,
\end{align*}
satisfying the following conditions:\\
\indent (i) Formal associativity, i.e.,
\begin{equation*}
\forall{p}\ge0,\quad\sum_{k+l=p}[((f\times_lg)\times_kh)-(f\times_k(g\times_lh))].
\end{equation*}
\indent (ii) $f\times_0g=fg$.\\
\indent (iii) $\frac{1}{2}(f\times_1g-g\times_1f)=\{f, g\}$, where \{ , \} is the Poisson bracket.\\
\indent (iv) Each map $\times_k:A\times{A}\to{A}$ should be a bidifferential operator.
\end{Def}

\begin{Def}[Formal Deformation of Poisson Bracket]
\ \\
A \emph{formal deformation} of the Poisson bracket on $A$ is askew-symmetric map
\begin{gather}
[,]:A\times{A}\to{A}[\hslash]\\
(f,g)\mapsto\sum_{k=0}^{\infty}T_k(f, g)\hslash^k
\end{gather}
satisfying:\\
\indent (i) the formal Jacobi identity, i.e.
\begin{equation}
\sum(\sum_{k+l=p}T_k(T_l(f, g), h))=0
\end{equation}
for all $p\le0$ and the outer sum is taken over the cyclic permutations of the set \{$f, g, h$\}.\\
\noindent (ii) $T_0(f, g)$=\{$f, g$\} where $\{,\}$ is the Poisson bracket.\\
\noindent (iii) Each map $T_k:A\times{A}\to{A}$ should be a bidifferential operator.
\end{Def}

\subsection{Quantized Tori}

Here we view the classical tori $\mathbb{T}^n=\mathbb{R}^n/\mathbb{Z}^n$ as Poisson manifolds, then following the deformation quantization we've introduced, we define the \emph{quantized tori} $T^n_\theta$ as $C^*$-algebras of self-adjoint operators.

\begin{Def}[Classical Tori]
\ \\
A classical $n$-torus $\mathbb{T}^n$ can be constructed by:
\begin{equation}
\forall{n}\in\mathbb{Z}^+,\quad\mathbb{T}^n=\mathbb{R}^n/\mathbb{Z}^n,
\end{equation}
which is actually an $n$-dimensional compact smooth manifold. Here, the collection of all smooth (infinitely differentiable) complex-valued functions on $\mathbb{T}^n$ is denoted as $C^{\infty}(\mathbb{T}^n)$. By Theorem 2.1, $C^{\infty}(\mathbb{T}^n)$ forms a $C^*$-algebra over $\mathbb{T}^n$.
\end{Def}

\begin{Def}[Poisson Structure on Classical Tori]
\ \\
Hence we can define a Poisson bracket by selecting a real skew-symmetric $n\times{n}$ matrix $\theta$:
\begin{equation}
\forall{f, g}\in{C}^{\infty}(\mathbb{T}^n),\quad\{f, g\}=\sum{\theta}_{jk}(\frac{\partial{f}}{\partial{x_j}})(\frac{\partial{g}}{\partial{x_k}}).
\end{equation}
\end{Def}

\begin{Def}[Canonical Commutation Relation]
\ \\
The deformation quantization of a classical torus is given by deforming the Poisson bracket to a one-parameter family of associate products $*_\hslash$, satisfying the \emph{canonical commutation relation}, i.e.,
\begin{equation}
\lim_{\hslash\to0}(\hat{f}*_\hslash\hat{g}-\hat{g}*_\hslash\hat{f})/i\hslash=\widehat{\{f, g\}}, \footnote{Throughout this paper, $i$ denote the imaginary unit $\sqrt{-1}$.}
\end{equation}
for all $f, g\in{C}^\infty(\mathbb{T}^n)$ and $\hat{\ }$ denote the Fourier transformation introduced below.
\end{Def}

\begin{Def}[Fourier Transformation]
\ \\
The \emph{Fourier transformations} $\hat{f}$ on $\mathbb{T}^n$ is defined by:
\begin{equation}
\hat{f}=\int\limits_{\mathbb{T}^n}\text{exp}(-2\pi{i}\vec{x}\cdot\vec{p})f(\vec{x})d\vec{x},
\end{equation}
for all ${f}\in{C^\infty}(\mathbb{T}^n)$.
\end{Def}

\begin{Def}[Schwartz Space]
\ \\
The \emph{Schwartz space} $S(\mathbb{Z}^n)$ over $\mathbb{Z}^n$ consists of the functions $f:\mathbb{Z}^n\to\mathbb{C}$ such that for $m\in\mathbb{Z}^n$, as $|m|\to\infty$, $f(m)$ decays to zero faster than any inverse power of $m$.
\end{Def}

\begin{Thm}
\ \\
A celebrated result is that the Fourier transformations carries $C^{\infty}(\mathbb{T}^n)$ onto $S(\mathbb{Z}^n)$. Where the pointwise multiplication on $C^{\infty}(\mathbb{T}^n)$ is  carried to convolution on $S(\mathbb{Z}^n)$.
\end{Thm}

\begin{Thm}[Fourier Transformation for Poisson Bracket]
\ \\
Fourier transformation from $C^\infty(\mathbb{T}^n)$ to $S(\mathbb{Z}^n)$ also carries the Poisson bracket to:
\begin{equation}
\widehat{\{f, g\}}(p)=-4\pi^2\sum_{q\in\mathbb{Z}^n}\hat{f}(q)\hat{g}(p-q)\gamma(q, p-q),
\end{equation}
for all $f, g\in{C^\infty}({\mathbb{T}}^n)$, $p\in\mathbb{Z}^n$ and
\begin{equation}
\gamma(p, q)=\sum\theta_{jk}p_jq_k.
\end{equation}
\end{Thm}

\begin{Def}[Star product]
\ \\
For $\hslash\in\mathbb{R}$, we define a family of \emph{bidifferential operator} $\sigma_\hslash$ on $\mathbb{Z}^n$ by
\begin{equation}
\times_\hslash:(p, q)\mapsto\exp(-\pi{i}\hslash\gamma(p, q)),
\end{equation}
and then the set
\begin{equation}
(\hat{f}*_\hslash\hat{g})(p)=\sum_{q\in\mathbb{Z}^n}\hat{f}(q)\hat{g}(p-q)(q\times_\hslash(p-q)).
\end{equation}
\end{Def}

\begin{Def}[Norm]
\ \\
Then for each $\hslash$, we can define the norm $||\ ||_\hslash$ on $S(\mathbb{Z}^n)$ to be the operator norm for the action of $S(\mathbb{Z}^N)$ on $l^2(\mathbb{Z}^n)$, i.e.,
\begin{equation}
||\hat{f}||_\hslash=\hat{f}*_\hslash\hat{f}.
\end{equation}
for all $f\in{C^\infty}({\mathbb{T}}^n)$.
\end{Def}

\begin{Def}[Involution]
\ \\
The \emph{Involution} on $S(\mathbb{Z}^n)$, independent of $\hslash$, is define by the complex conjugation on $C^\infty(\mathbb{T})^n$ such that,
\begin{equation}
(\hat{f})^*(p)=\widehat{(\overline{f})}(-p),
\end{equation}
for all $f\in{C^\infty}({\mathbb{T}}^n)$, $p\in\mathbb{Z}^n$ and $\overline{f}$ is the complex conjugation of $f$.
\end{Def}

\begin{Def}
\ \\
By pulling back from $S(\mathbb{Z}^n)$ the product $*_\hslash$, the involution $^*$ and the induced norm $||\ ||_\hslash$ with the inverse Fourier transformation, we have the star product, involution, and norm, defined on $C^\infty(\mathbb{T}^n)$ but denoted by the same symbol for convenience. Then define $C_\hslash$ to be the set $C^\infty(\mathbb{T}^n)$ equipped with $*_\hslash$, $^*$ and ${||\ ||}_\hslash$.
\end{Def}

\begin{Thm}
\ \\
The completion of $C_\hslash$ form a continuous field of $C^*$-algebras and
\begin{equation*}
\lim_{\hslash\to0}||(f*_\hslash{g}-g*_\hslash{f})/i\hslash-\{f, g\}||_\hslash=0,
\end{equation*}
for all ${f, g}\in{C}^\infty(\mathbb{T}^n)$.
\end{Thm}

\noindent Hence we have $C_0$ is just $C^\infty(\mathbb{T}^n)$ with the ordinary pointwise product. By definition, we can see $C_\hslash$ forms a deformation quantization of $C^\infty(\mathbb{T}^n)$ in the direction of the Poisson Bracket respect to $\theta$.

\begin{Def}[Quantized Tori]
\ \\
Let $T^n_\theta$ denote the algebra respect to $C_1$, we call it as a \emph{quantized }$n$-\emph{torus} or a \emph{non-commutative} $n$-\emph{torus}. And let $\overline{T^n_\theta}$ denote the norm completion of $T^n_\theta$, which is a non-commutative $C^*$-algebra.
\end{Def}

\begin{Def}[Translations on Quantized Tori]
\ \\
The \emph{translation} on $\mathbb{T}^n$ forms an action on $C^\infty(\mathbb{T}^n)$, which is also an action by a continuous $C^*$-algebra automorphism on $T^n_\theta$ with respect to $*_\hslash$. Hence the translation on $\mathbb{T}^n$ gives a dual action $\lambda$ on $T^n_\theta$. On the Schwartz space $S(\mathbb{Z}^n)$, this action is given by:
\begin{equation*}
(\lambda_{t}(\hat{f}))(p)=\text{exp}(2\pi{i}t\cdot{p})\hat{f}(p),
\end{equation*}
for all $t\in\mathbb{T}^n$, $p\in\mathbb{Z}^n$, and $f\in{T^n_\theta}$ and $\lambda_{t}$ denote the translation on $\mathbb{T}^n$ given by $t$.
\end{Def}

\begin{Def}[Unitary Generators of Quantized Tori]
\ \\
For each $p\in\mathbb{Z}^n$, the function
\begin{align}
\mu: t&\mapsto{e}^{2\pi{i}t\cdot{p}}
\end{align}
correspond to a \emph{unitary operator} $U_{p}$ in $\overline{T}^n_\theta$. And $\vec{p}\mapsto{U_{\vec{p}}}$ is then a unitary representation of $\mathbb{Z}^n$. Thus $\overline{T}^n_\theta$ can be considered as the $C^*$-algebra generated by these unitary representations. Let $U_1,\dots,U_n$ denote the unitary operators associate with the standard basis of $\mathbb{T}^n$, then these unitary operators generate $\overline{T}^n_\theta$ and satisfy
\begin{equation}
{U_k}{U_j}=e^{2\pi{i}\theta_{jk}}{U_j}{U_k}
\end{equation}
for any $j, k\in\{1,\dots,n\}$.
\end{Def}

\subsection{Topological structure of Quantized Tori}

To study the non-commutative geometry of quantized tori, we have to construct projective modules on $T^n_\theta$. In the context of commutative geometry, Serre-Swan Theorem provides an bijective correspondence between the sections of complex smooth vector bundles and finitely generated projective modules of the $C^*$-algebra of continuous functions over the space. Since the non-commutative tori $T^n_\theta$ is constructed as a space of functions with a non-commutative $*$ product, then we can consider the projective modules of $T^n_\theta$ as an non-commutative analogue to the (space of sections of) vector bundles.

\begin{Def}[Projective Modules]
\ \\
The \emph{projective} $T^n_\theta$-\emph{module} $\Xi$ is a space such that
\begin{equation}
\beta(f)\xi\in\Xi
\end{equation}
for $\beta$ is a linear action on $T^n_\theta$, $f\in{T}^n_\theta$, $\xi\in\Xi$and the product is taken in pointwise multiplication. In addition
it is given by the form
\begin{equation}
\Xi=p\otimes^m{T^n_\theta}
\end{equation}
for $p\in{{Mat}_m(T^n_\theta)}$.
\end{Def}

\noindent Now we would introduce the general form of the projective modules on quantized tori, which is a result given by Riffel in [R2], and is cited and explained in section 5.6 of [KS].

\begin{Thm}[Projective Modules on Quantized Tori]
\ \\
Given a quantized tous $T^n_\theta$. The \emph{projective} $A_\theta$-\emph{modules} are given by
\begin{equation}
\Xi\cong{S}(\mathbb{R}^p\times\mathbb{Z}^q\times{F})
\end{equation}
for all $p, q\in\mathbb{Z}^n$ s.t. $2p+q=n$. Here $S(G)$ denote the Schwartz space over $G$ and $F$ is a finite group. In particular, the trivial bundle $E$ over $T^n_\theta$ is given by $\Gamma(E)\cong{S(\mathbb{Z}^n)}$.
\label{module}
\end{Thm}

\subsection{Smooth structure on Quantized Tori}

Recall that the \emph{translation} on $\mathbb{T}^n$ forms an dual action $\lambda$ on $T^n_\theta$, which is given by:
\begin{equation*}
(\lambda_{t}(\hat{f}))(p)=\text{exp}(2\pi{i}t\cdot{p})\hat{f}(p).
\end{equation*}
We can then develop the concept of differentiation on quantized tori by considering the infinitesimal generator of the translation on $\overline{T}^n_\theta$.

\begin{Def}[Directional Differentiation on Quantized Tori]
\ \\
Let $\delta_k$ denote the differentiation on $T^n_\theta$ in the $k^{th}$ direction of $\mathbb{T}^n$, which is given by:
\begin{equation}
(\delta_k(\hat{f}))(p)=2\pi{i}p_k\hat{f}(p),
\end{equation}
where $f\in{T^n_\theta}$ and $p\in\mathbb{T}^n$.\\
\\
Each $\delta_k$ is a $*$-derivation of $\mathbb{T}$, which satisfies:
\begin{gather}
(\delta_k(a))^*=\delta_k(a^*)\\
\delta_k(ab)=\delta_k(a)b+a\delta_k(b)
\end{gather}
for all $a, b\in{T^n_\theta}$.
\end{Def}

\begin{Def}[Derivations on Quantized Tori]
\ \\
Let $T^n_\theta$ denote a non-commutative n-tori. Then let $L$ denote all the derivatives on $T^n_\theta$, i.e.,
\begin{equation}
L=\{\delta|\delta:T^n_\theta\to{T^n_\theta}\ \text{is linear and}\ \delta(\varphi\psi)=\delta(\varphi)\psi+\psi\delta(\varphi)\}
\end{equation}
and its complexification is defined by $L\oplus{i}L$.
\end{Def}

\begin{Cor}
\ \\
The derivations $L$ of a non-commutative torus $T^n_\theta$ is isomorphic to $\mathbb{R}^n$ in the sense of Lie algebra. Hence we have $L\oplus{i}L\cong\mathbb{C}^n$.
\end{Cor}

\section{Quantized Theta vectors}

In this section, we introduce the \emph{theta vectors} defined by Schwarz in [S] as holomorphic vectors in projective modules of quantum tori. Following the procedure in [E], we start with the holomorphic condition on quantized tori and then define Schwarz's theta vector as a holomorphic vector in $S(\mathbb{R})$. We would start from complex structure on non-commutative tori, then we study the connections and holomorphic conditions for the projective modules, and finally we define a theta vector on non-commutative tori.

\subsection{Complex structure on Quantized Tori}

In this subsection, we give the definitions of complex structure of non-commutative manifold which actually are directly induced by the definitions on commutative complex manifolds.

\begin{Def}[Complexification of Vector Spaces]
\ \\
A \emph{complexification} of a real vector space $V$ is given by
\begin{equation}
V\otimes\mathbb{C}
\end{equation}
where $\otimes$ is the tensor product over real spaces, or equivalent
\begin{equation}
V\oplus{i}V
\end{equation}
\end{Def}

\noindent Given $L$ is a Lie algebra isomorphic to $\mathbb{R}^n$. As a vector space, its complixification is given by
\begin{equation}
L\oplus{i}L\cong\mathbb{C}^n.
\end{equation}
However, if we view $L\cong\mathbb{R}^n$ as the derivatives of $T^n_\theta$, then the complex structure of $L$ must be modified so that the complex deriavatives on $T^n_\theta$ forms a subspace of $L\oplus{iL}$ isomorphic to $\mathbb{C}^m$ for $n=2m$ since every complex directional derivative is give by two real derivatives.

\begin{Def}[Complexification of Manifolds]
\ \\
For $L$ denote the (real) derivatives of $T^n_\theta$, we separate the the directions of the derivatives in two collections so that
\begin{align}
L&=\text{span}\{\delta_1,\dots,\delta_n\}\\
&=\mathbb{R}(\frac{\partial}{\partial{x_j}}, \frac{\partial}{\partial{y_k}}),
\end{align}
for $j, k\in$\{1,\dots, m\}. The \emph{complex structure} on $L$ is then given by
\begin{align}
L\oplus{i}L&=\mathbb{C}(\frac{\partial}{\partial{x_j}}, \frac{\partial}{\partial{y_k}})\\
&=\mathbb{C}(\frac{\partial}{\partial{z_k}}, \frac{\partial}{\partial{\overline{z}_k}})\\
&=L^{1, 0}\oplus{L}^{0, 1}
\end{align}
where
\begin{gather}
\frac{\partial}{\partial{z_k}}=\frac{1}{2}(\frac{\partial}{\partial{x_k}}-\sqrt{-1}\frac{\partial}{\partial{y_k}}),\\
\frac{\partial}{\partial{\overline{z}_k}}=\frac{1}{2}(\frac{\partial}{\partial{x_k}}+\sqrt{-1}\frac{\partial}{\partial{y_k}}).
\end{gather}
and the bases of the two complex conjugate subspace $L^{1, 0}$ and $L^{0, 1}$, $\{\delta_1,\dots,\delta_n\}$ and $\{\tilde{\delta}_1,\dots,\tilde{\delta}_n\}$ respectively. And $L^{1, 0}$ is called the holomorphic bundle and $L^{0, 1}$ is called the antiholomorphic bundle.
\end{Def}

\begin{Thm}[Coordinate Transformation]
\ \\
By construction there exists a complex $n\times{n}$ matrix $h$ which can express $\tilde{\delta}_\alpha$ with $\delta_\alpha$ by
\begin{equation}
\tilde{\delta}_\alpha=h_{\alpha}^\beta\delta_\beta.
\end{equation}
and we can write it as
\begin{equation}
h=
\begin{bmatrix}
1&0\\
0&T
\end{bmatrix}
\end{equation}
where $T$ is an $m\times{m}$ invertible matrix for $n=2m$.
\end{Thm}

\begin{Def}[Complex Structure on Quantized Tori]
\ \\
Let $T^n_\theta$ denote a non-commutative torus, we say $T^n_\theta$ is equipped with a \emph{complex structure} if the Lie algebra $L\cong\mathbb{R}^n$ acting on it is equipped with a complex structure as we defined above.
\end{Def}

\noindent Given the vector bundles on a manifold, it is natural to equip it with an \emph{Hermitian metric} which enable us to have a continuous inner product on the vector bundles.

\begin{Def}[Hermitian Metrics]
\ \\
Let $\mathfrak{A}$ denote a unital $C^*$-algebra, and $\Xi$ a projective right $\mathfrak{A}$-module. A \emph{Hermitian metric} on $\Xi$ is a bi-additive $\mathfrak{A}$-value function ${<|>}_{\mathfrak{A}}$ on $\Xi\times\Xi$ which satisfy
\begin{align}
&{\langle\xi|\eta{a}\rangle}_{\mathfrak{A}}={\langle\xi|\eta\rangle}_{\mathfrak{A}}a;\\
&{\langle\xi|\eta\rangle}^*_{\mathfrak{A}}={\langle\eta|\xi\rangle}^*_{\mathfrak{A}};\\
&{\langle\xi|\xi\rangle}_{\mathfrak{A}}\ge0.
\end{align}
for all $\xi, \eta\in\Xi$ and $a\in\mathfrak{A}$. And for any linear map $\phi:\Xi\to\mathfrak{A}$ such that $\phi(\xi{a})=\phi(\xi)a$, there must exists an $\eta\in\Xi$ for which
\begin{equation}
\forall\xi\in\Xi,\ \phi(\xi)={\langle\eta|\xi\rangle}_{\mathfrak{A}}.
\end{equation}
\end{Def}

\begin{Thm}
\ \\
Any (finitely generated) projective module $\Xi$ over $T^n_\theta$ can always equip with a $T^n_\theta$-valued Hermitian metric ${\langle|\rangle}_{T^n_\theta}$. And it induces a norm over $\Xi$:
\begin{equation}
||\xi||={||{\langle\xi|\xi\rangle}_{T^n_\theta}||}^{\frac{1}{2}}
\end{equation}
for all $\xi\in\Xi$.
\end{Thm}

\begin{proof}
Given a projective $T^n_\theta$-module, it can be viewed as a summand of free modules. Hence the standard Hermitian metric on a free module would give a induced Hermitian module over the projective module.
\end{proof}

\subsection{Geometrical structure of Quantized Tori}
Given the differentials on the quantized tori $\delta_X$ defined, we can then study the geometric structure of the quantized tori, such as connections and curvatures.

\begin{Def}[Connections]
\ \\
Let $\Xi$ denote a projective $T^n_\theta$-module, then a \emph{connection} for $\Xi$ is a linear map
\begin{equation}
\nabla:\Gamma(\Xi)\mapsto\Gamma((L+iL)^*\otimes\Xi)
\end{equation}
satisfying the Leibnitz rule
\begin{equation}
\triangledown_\alpha(\xi{f})=(\triangledown_\alpha\xi)f+\xi\delta_\alpha(f)
\end{equation}
for all $\alpha\in{L+iL}$, $\xi\in\Xi$, and $f\in{T^n_\theta}$, and $(L+iL)^*$ is the dual space of $(L+iL)$.
\end{Def}

\begin{Def}[Curvatures]
\ \\
The \emph{curvature} $R^\nabla$ of a curvature $\nabla$ is simply defined by
\begin{equation}
{R}^\nabla(\alpha, \beta)=[\triangledown_\alpha, \triangledown_\beta].
\end{equation}
for all ${\alpha, \beta}\in{L+iL}$.
\end{Def}

\noindent By defining the connections, we can then define the complex structure on $T^n_\theta$-modules which is compatible to the geometric structure on $T^n_\theta$.

\begin{Def}[Complex Structure on $T^n_\theta$-modules]
\ \\
Let $\Xi$ denote a $T^n_\theta$-module and. Then the \emph{complex structure} on $\Xi$ can be defined by a collection of complex linear operators such that,
\begin{gather}
\tilde{\triangledown}_{\alpha}(\xi{f})=(\tilde\triangledown_{\alpha}\xi)f+\xi\tilde{\delta}_{\alpha}f,\\
[\tilde{\triangledown}_\alpha, \tilde{\triangledown}_\beta]=0,
\end{gather}
for $\xi\in\Xi$ and $f\in{T}^n_\theta$. And $\triangledown$ is
\begin{equation}
\tilde{\triangledown}_{\alpha}=h_\alpha^\beta\triangledown_\beta
\end{equation}
then the two conditions are satisfied.
\end{Def}

\subsection{Theta Vectors on Quantized Tori}

Ordinary theta functions can be considered as holomorphic sections of line bundles over (classical) tori. As an analogue in non-commutative geometry, Schwarz construct theta vectors as holomorphic elements of projective modules over non-commutative tori.

\begin{Def}[Holomorphic Vectors]
\ \\
A vector $f\in\Xi$ is \emph{holomorphic} if
\begin{equation}
\tilde{\triangledown}_\alpha{f}=0
\end{equation}
for all $\alpha\in\{1,\dots,m\}$. In another words, $f$ vanishes in the anti-holomorphic bundle.
\end{Def}

\noindent As Theorem \ref{module} stated, the space of Schwartz functions $S(\mathbb{R}^m)$ is a finitely generated projective module over the non-commutative torus $T^n_\theta$. In the following theorem, the formulae of connections of $S(\mathbb{R}^m)$ is given in [S] and [E].

\begin{Thm}
\ \\
If the projective modules over $T^n_\theta$ is given by $S(\mathbb{R}^m)$ where $n=2m$, we can choose a connection $\nabla$ on $S(\mathbb{R}^n)$ is defined by:
\begin{align}
&\triangledown_\alpha=\frac{\partial}{\partial{x^\alpha}}\\
&\triangledown_{\alpha+m}=-2\pi{i}\sigma_\alpha{x}^\alpha
\end{align}
for all $\alpha\in\{1,\dots,m\}$, $\sigma_\alpha\in\mathbb{R}$. In the equations, $x^\alpha$ denote the coordinates of $\mathbb{R}^n$ and the repeated indices are not summed.
\end{Thm}

\begin{Thm}[Equivalence Condition for Holomorphicity]
\ \\
An equivalent condition for a vector $f$ to be holomorphic is
\begin{equation}
(\frac{\partial}{\partial{x^\alpha}}-\sum_\beta2\pi{i}\tau_{\alpha\beta})f=0.
\end{equation}
where $\tau$ is a $m\times{m}$ matrix given by $\tau_{\alpha\beta}=T_{\alpha\beta}\sigma_{\beta}$ for $\alpha, \beta\in\{1,\dots,m\}$. Recall that $
h=
\begin{bmatrix}
1&0\\
0&T
\end{bmatrix}
$
is the matrix transform the coordinates form the holomorphic bundle $L^{1, 0}$ to the antiholomorphic bundle $L^{0, 1}$.
\end{Thm}

\begin{Def}[Theta Vectors]
\ \\
The vectors satisfying the holomorphicity condition are called the \emph{theta vectors}. Up to a constant, a \emph{theta vector} can be written in the form:
\begin{equation}
\vartheta(x^1,\dots,x^m)=e^{\pi{i}x^\alpha{\tau_{\alpha\beta}}x^\beta}.
\end{equation}
where the repeated indices are not summed.
\label{theta vector}
\end{Def}

\section{KdV equations and Classical Theta Function}
The KdV equation, firstly formulated by Korteweg and Vries, governs the propagation of a solitary wave in a shallow canal. It is a very first example of a soliton equations and now it has been generalized into KdV hierarchy and KP hierarchy. Surprisingly, Riemann's theta functions are found to be able to give exact solutions of KdV and other soliton equations via algebraic geometric methods. In this section we would introduce the simplest case of KdV with one soliton and theta functions of genus 1.

\subsection{KdV equation}

\begin{Def}[KdV equation]
\ \\
The \emph{KdV equation} in the simplest form is given by
\begin{equation}
u_t=-6uu_{y}-u_{yyy}
\end{equation}
for $u(x, t)\in{C^\infty}(\mathbb{R}^2)$. In the context of the solitary wave in a shallow canal, the variable $t$ denote the time and the variable $x$ denote the space coordinate along the canal, while the function $u(x, t)$ represents of the elevation of the fluid above the bottom of the canal.
\end{Def}

\begin{Def}[Soliton Solution]
\ \\
For $u(x,t)$ is a \emph{soliton solution} of a KdV equation, it represents a solitary wave which is translation invariant. Hence it can be written in the form
\begin{equation}
u(x, t)=f(x-ct),
\end{equation}
where $c$ is some constant. Then, the KdV equation is reduced to
\begin{equation}
f'''=-6ff'+cf'.
\label{KdV}
\end{equation}
\end{Def}

\begin{Thm}
\ \\
With the condition of $u$ fast decreasing at $\pm\infty$, the solutions of (\ref{KdV}), there are solutions given by
\begin{equation}
f(x)=\frac{1}{2}c\cdot\text{sech}^2(\frac{1}{2}\sqrt{c}\ x)
\end{equation}
where $a=x-ct$ and sech$(x)=2{(e^x+e^{-x})}{^{-1}}$. Hence we have a family of soliton solutions for KdV
\begin{equation}
u(x, t)=\frac{1}{2}c\cdot\text{sech}^2(\frac{1}{2}\sqrt{c}\ (x-ct)+x_0)-\frac{1}{6}c
\label{KdV-s}
\end{equation}
\end{Thm}

\begin{Def}[Weierstrass $\wp$ function]
\ \\
The Weierstrass $\wp$ function is defined by
\begin{equation}
\wp(z)=\frac{1}{z^2}+\sum_{\omega\in\Lambda\backslash0}(\frac{1}{{(z-\omega)}^2}-\frac{1}{\omega^2})
\end{equation}
where $\Lambda$ is a lattice in $\mathbb{C}$ generated by two independent vector $\omega_1$ and $\omega_2$.
\end{Def}

\noindent Weierstrass $\wp$ function satisfies the following relation
\begin{equation}
\wp'''=12\wp\wp'
\label{wp}
\end{equation}

\noindent Since the equations (\ref{KdV}) and (\ref{wp}) looks quite similar, we seek to find a solution of (\ref{KdV}) by $\wp$, and immediately found that
\begin{equation}
u_{\Lambda}(x, t)=-{1}{2}c\cdot\wp(\frac{1}{2}\sqrt{c}(x-ct)+x_0)-\frac{1}{6}c
\end{equation}
also gives a family of soliton solutions of KdV equation.

\subsection{Hirota Bilinear relation and theta functions}

Here we introduce Hirota's method of solving KdV, which can be extended to $n$-soliton cases. We focus on the case of 1-soliton case which can be considered as a computational trick relating theta functions and KdV equations.\\

\noindent Hirota's idea is to use the substitution
\begin{equation}
u=2\frac{d^2}{dx^2}\log{f}
\label{H-trick}
\end{equation}
in the computation.\\

\noindent First substitute $u=d_x$ in the KdV, integrating with respect to $x$, one get
\begin{equation}
g_t+3g_x^2+g_{xxx}=0,
\end{equation}
then, assign $g$ with $2\frac{d}{dx}\log{f}$, the KdV is finally translated into the \emph{Hirota bilinear form}.

\begin{Def}[Hirota Bilinear Form]
\ \\
The \emph{Hirota bilinear form} of KdV equation is given by
\begin{equation}
ff_{xxxx}+3f^2_{xx}-4f_xf_{xxx}-4f_xf_{xxx}-f_tf_x+ff_{tx}=0.
\end{equation}
\end{Def}

\noindent In this framework, by setting 
\begin{equation}
f=1+exp(\alpha{x}-\alpha^3t+x_0),
\end{equation}
we can see the soliton solutions in the form of (\ref{KdV-s}). In fact (\ref{KdV-s}) can be derived when $\alpha=\sqrt{c}$ and define $u$ by (\ref{H-trick}).\\

\begin{Thm}
\ \\
According to our the definition of single variable classical theta function (Def. \ref{theta}), we have the following identity
\begin{equation}
\wp(z, \tau)=-\frac{d^2}{dz^2}\log\vartheta(z+\frac{1}{2}(1+\tau), \tau)+k
\end{equation}
by setting $\wp$ on the lattice $\Lambda$ generated by 1 and $\tau\in\mathbb{H}$. Hence,
\begin{equation}
u(x, t)=2\frac{d^2}{dx^2}\log{\vartheta}(\alpha{x}+\gamma{t}+x_0, \tau)
\end{equation}
for $\alpha, \gamma, x_0\in\mathbb{Z}$.
\label{soliton}
\end{Thm}

\noindent Hence, we can see, Hirota's method related Weierstrass $\wp$ function to the theta function. Since theta functions can be extended to higher dimensions by Riemann's theta function, Hirota's approach expresses KdV in terms of theta functions instead of Weistrass $\wp$ functions. Where theta functions can be generalized into Riemann's theta functions to solve KdV equations in higher dimensions.

\section{KdV equation and Theta Vector on Non-commutative 2-tori}

In this section, we would study the relationship between KdV equations and theta vectors. Here we only focused on the case of non-commutative 2-tori. First we produce the explicit formulae for the projective modules, connections, holomorphic conditions and theta vectors on 2-tori. Then we showed that quantum theta vector satisfies Hirota's bilinear form.

\subsection{Quantized 2-Tori}

In this subsection, we introduce the explicit formulae for the unitary generators, projective modules, connections and curvature in the particular case of $T^2_\theta$ which are are calculated out in [C], [R], [PS] and [KS] in detail. These results would serve as the framework for our calculation in subsection 6.5.
 
\begin{Def}[Quantized 2-Tori]
\ \\
The special case of a \emph{quantized} 2-\emph{torus} $T^2_\theta$, where $\theta\in\mathbb{R}$, is given by a deformation quantization of ordinary 2-torus $\mathbb{T}^2$. And here the $2\times2$ skew-symmetric real matrix $\theta$ is only defined by a real number. Hence for convenience, we denote the skew symmetric matrix as
\begin{equation}
\begin{bmatrix}
\ \ 0&\theta\\
-\theta&0
\end{bmatrix}
.
\end{equation}
\end{Def}

\begin{Def}[Unitary Generators of Quantized 2-tori]
\ \\
As introduced in [PS], the $C^*$-algebra $\overline{T^2_\theta}$ consists of formal linear combinations
\begin{equation}
\sum_{(n,m)\in\mathbb{Z}^2}a_{nm}U^nV^m,
\end{equation}
where the coefficient function $(n, m)\mapsto{a}_{nm}$ rapidly decreasing at infinity, and $U$, $V$ are two unitary operators in $\overline{T^2_\theta}$ such that
\begin{equation}
UV=\lambda{VU}
\end{equation}
for $\lambda=e^{2\pi{i}\theta}$.
\end{Def}

\begin{Def}[Projective Modules on $T^2_\theta$]
\ \\
According to Theorem 4.5, the (finitely generated) projective modules on $T^2_\theta$ are the trivial case $S(\mathbb{Z}^2)$ and $S(\mathbb{R})$. Where $S(R)$ is equipped with the following right action of $T^n_\theta$:
\begin{gather}
(U\xi)(p)=\xi(p+\kappa),\\
({V}\xi)(p)=\exp(2\pi{i}\tilde{\kappa}p)\xi(p)
\end{gather}
for all $\psi\in{S(\mathbb{R})}$ and $\kappa, \tilde{\kappa}\in\mathbb{R}$ such that $\kappa\tilde{\kappa}=\theta$.
\end{Def}

\begin{Def}[Derivatives on  $T^2_\theta$]
\ \\
The \emph{derivatives} are given by
\begin{equation}
L=\text{span}\{\delta_1, \delta_2\}\cong\mathbb{R}^2
\end{equation}
where $\{\delta_1, \delta_2\}$ is the standard basis of $L$ corresponding to the two directions on $\mathbb{T}^2$.
\end{Def}

\begin{Def}[Connections on $T^2_\theta$]
\ \\
Let $\Xi$ denote a $T^2_\theta$-projective module. Then a \emph{connection} $\nabla$ can be defined by its values:
\begin{align}
(\triangledown_1\xi)(x)&=(\frac{d\xi}{dx})(x)\\
(\triangledown_2\xi)(x)&=(\frac{2\pi{i}x}{\theta})\xi(x)
\end{align}
where $\xi\in\Xi$. Where this result is calculated out in [R].
\end{Def}

\begin{Def}[Curvature on $T^2_\theta$]
\ \\
The \emph{curvature} $R^\nabla$ of a connection $\nabla$ on $T^2_\theta$ is then given by:
\begin{equation}
R^\nabla(\vec{1}, \vec{2})=[\triangledown_1, \triangledown_2]=\frac{2\pi{i}}{\theta}\mathbb{I}
\end{equation}
where $\mathbb{I}$ is the identity operator on $S(\mathbb{R})$.
\end{Def}

\subsection{Theta Vector on $T^2_\theta$}

Recall Definition \ref{theta vector}, the theta vector on $T^2_\theta$ is in the form of 
\begin{equation*}
\vartheta(x_1,\dots, x_m)=e^{2\pi{i}x^\alpha\tau_{\alpha\beta}{x}^{\beta}}.
\end{equation*}
The $m\times{m}$ matrix $\tau$ given by $\tau_{\alpha\beta}={T}_{\alpha\beta}\sigma_\beta$ is degenerated to a complex number since $n=2$ and $2m=n$ where $\sigma_\beta\in\mathbb{R}$ is given by $\sigma_\beta=\theta$ and $\tau\in\mathbb{H}$ for $\mathbb{H}$ denote the upper half space of $\mathbb{C}$. Then the summation $x^\alpha{\tau}_{\alpha\beta}{x}^{\beta}$ for $\alpha\in\{x_1,\dots,x_m\}$ is degenerated to a $x\tau{x}=x^2{\tau}$. Hence we have the following definition.

\begin{Def}[Theta Vector on $T^2_\theta$]
\ \\
The \emph{theta vector} on $T^2_\theta$ is given by
\begin{equation}
\vartheta(x)=e^{2\pi{i}x^2\tau}\in{S(\mathbb{R})}
\end{equation}
where $x\in\mathbb{R}$ and $\tau\in\mathbb{H}$.
\end{Def}

\begin{Def}[Partial Derivatives of Theta Vectors]
\ \\
The partial derivatives of theta vectors along the standard basis of  $T^2_\theta$ are given by:
\begin{align}
\triangledown_{1}\vartheta&=\frac{d}{dx}e^{2\pi{i}x^2\tau}=(2\pi{i}x\tau)\vartheta(x)\\
\triangledown_{2}\vartheta&=(\frac{2\pi{i}x}{\theta})\vartheta(x)
\end{align}
\end{Def}

\subsection{Quantum KdV equations}

\begin{Def}[Quantum KdV equation on $T^2_\theta$]
\ \\
The quantum KdV equations are defined as:
\begin{equation}
\xi_t=\xi\xi_{y}+6\xi_{yyy}
\end{equation}
where $\xi$ is a smooth section over $T^2_\theta$ (i.e. $\xi\in{S(\mathbb{R}}$)), and $\xi_y=\triangledown_{i}\xi$, $\xi_t=\triangledown_{j}\xi$ for $i$, $j$ denote independent directions.
\end{Def}

\begin{Def}[Hirota Bilinear Relation]
\ \\
By the same procedure in section 5, one can show that any function satisfied the \emph{Hirota bilinear relation} would be solutions of the quantized KdV equation on $T^2_\theta$. Where the Hirota bilinear form is given by:
\begin{equation}
\eta\eta_{yyyy}+3\eta^2_{yy}-4\eta_{y}\eta_{yyy}-\eta_t\eta_y+\eta\eta_{ty}=0
\end{equation}
for $\eta\in{S(\mathbb{R})}$.
\end{Def}

\begin{Thm}[Theta Vectors on $T^2_\theta$ satisfy Hirota Bilinear Form]
Given the theta vector
\begin{equation}
\vartheta(y)=e^{2\pi{i}x^2\tau},
\end{equation}
and let
\begin{equation}
\vartheta_y=\frac{\partial}{\partial{y}}\vartheta=\triangledown_{2}\vartheta\ \text{and}\ \vartheta_t=\frac{\partial}{\partial{t}}\vartheta=\triangledown_{1}\vartheta
\end{equation}
we have,
\begin{equation}
\vartheta\vartheta_{yyyy}+3\vartheta^2_{yy}-4\vartheta_{y}\vartheta_{yyy}-\vartheta_t\vartheta_y+\vartheta\vartheta_{ty}=0.
\end{equation}
\end{Thm}

\begin{proof}
\ \\
By applying the formulae (70, 71), we have:
\begin{equation}
\theta_y=(\frac{2\pi{i}x}{\theta})\vartheta(x)\ \text{and}\ \theta_t=(2\pi{i}x\tau)\vartheta(x),
\end{equation}
then
\begin{gather}
\vartheta\vartheta_{yyyy}=\vartheta_{yy}\vartheta_{yy}=\vartheta_{y}\vartheta_{yyy}={(\frac{2\pi{i}x}{\theta})}^4{\vartheta}^2,\\
\theta_t\theta_y=\theta\theta_{ty}=(\frac{2\pi{i}x}{\theta})^2\tau\vartheta^2,
\end{gather}
hence
\begin{equation}
\vartheta\vartheta_{yyyy}+3\vartheta^2_{yy}-4\vartheta_{y}\vartheta_{yyy}-\vartheta_t\vartheta_y+\vartheta\vartheta_{ty}=0.
\end{equation}
\end{proof}

\begin{Cor}
\ \\
Then the vector
\begin{equation}
\psi(x, t)=2\frac{d^2}{dx^2}\log{\vartheta}(ax+bt+c, \tau)
\end{equation}
for all $a, b, c\in\mathbb{R}$ forms a family of soliton equations of quantized KdV equation.
\end{Cor}

\begin{Rem}
If we alternatively let
\begin{equation}
\vartheta_y=\frac{\partial}{\partial{y}}\vartheta=\triangledown_{1}\vartheta\ \text{and}\ \vartheta_t=\frac{\partial}{\partial{t}}\vartheta=\triangledown_{2}\vartheta
\end{equation}
then we have,
\begin{gather}
\vartheta\vartheta_{yyyy}+3\vartheta^2_{yy}-4\vartheta_{y}\vartheta_{yyy}=4{(2\pi{i}\tau)}^2\vartheta^2\\
-\vartheta_t\vartheta_y+\vartheta\vartheta_{ty}=\frac{2\pi{i}}{\theta}\vartheta^2
\end{gather}
If we assume Hirota bilinear relation is satisfied, we then have
\begin{equation}
\tau^2=\frac{i\theta}{8\pi}
\end{equation}
which is unsatisfiable. Hence we can say that we've made the right choice for the directions of $t$ and $y$. 
\end{Rem}

\section{Bibliography}
\noindent [C] Alain Connes, \emph{C*-algebres et Geometrie Differentielle}, C. R. Acad. Sci. Paris 290 (1980), 599-604.\\
\noindent [M1] Yuri I. Manin, \emph{Quantized theta function}. Progress of Theoretical Physics Supplement No. 102, 1990, 219-228\\
\noindent[M2] Yuri I. Manin, \emph{Theta Functions, Quantum Tori and Heisenberg Groups}. https://arxiv.org/abs/math/0011197v1\\
\noindent[S] Albert Schwartz, \emph{Theta-functions on noncommutative tori}, Letters in Mathematical Physics 58: 81-90, 2001. https://arxiv.org/abs/math/0107186v1\\
\noindent[PS] A. Polishchuk, Albert Schwartz, \emph{Categories of Holomorphic Vector Bundles on Noncommutative Two-Tori} https://arxiv.org/pdf/math/0211262.pdf\\
\noindent [KS] Anatoly Konechny, Albert Schwartz, \emph{Introduction to M(atrix) theory and noncommutative geometry} https://arxiv.org/abs/hep-th/0012145v3\\
\noindent[P] Emma Previato, (2018) \emph{Complex algebraic geometry applied to integrable dynamics: Concrete examples and open problems.} In: Kielanowski P., Odzijewicz A., Previato E. (eds) Geometric Methods in Physics XXXV. Trends in Mathematics. Birkhäuser, pp 277-288\\
\noindent[GH] Fritz Gesztesy, Helge Holden, \emph{Soliton Equations and their Algebro-Geometric Solutions, Volume I: (1+1)- dimensional Continuous Models}, Cambridge University Press, Cambridge studies in advanced mathematics 79.\\
\noindent[R1] Marc A. Rieffel, \emph{Non-commutative Tori -- A Case Study of Non-commutative Differentiable Manifolds}, Contemporary Mathematics 105 (1990) 191-211.\\
\noindent[R2] Marc A. Rieffel, \emph{Projective Modules over Higher-dimensional Non-commutative Tori}, Can. J. Math., Vol. XL, No. 2, 1998, pp. 257-338\\
\noindent[L] N. P. Landsman, \emph{Lecture Notes on} $C^*$-\emph{algebra, Hilbert }$C^*$-\emph{modules, and Quantum Mechanics}, https://arxiv.org/abs/math-ph/9807030.\\
\noindent[E] Ee Chang-Young, Hoil Kim, \emph{Theta Vectors and Quantum Theta Functions}, Journal of Physics A: Mathematical and General, Volume 38, Number 19.

\end{document}